\documentclass{article}
\usepackage{fullpage}
\usepackage{amsfonts}
\usepackage{graphicx}
\usepackage{amsmath,amssymb,amsthm}
\usepackage{authblk}
\usepackage{url}
\newcommand{\REV}[1]{\ensuremath{\overline{#1}}}
\newcommand{\RLCSA}{\ensuremath{\mathsf{RLBWT}}}
\newcommand{\ST}{\ensuremath{\mathsf{ST}}}

\newcommand{\SA}{\ensuremath{\mathsf{SA}}}

\newcommand{\BWT}{\ensuremath{\mathsf{BWT}}}
\newcommand{\CDAWG}{\ensuremath{\mathsf{CDAWG}}}
\newcommand\SP[1]{\mathtt{sp}(#1)}
\newcommand\EP[1]{\mathtt{ep}(#1)} 
\newcommand\INTERVAL[1]{\mathtt{range}(#1)}
\newcommand{\INTERVALFUNCTION}{\ensuremath{\mathbb{I}}}

\newcommand{\occ}
  {\ensuremath{\mathsf{occ}}}
\newtheorem{definition}{Definition}
\newtheorem{lemma}{Lemma} 
\newtheorem{theorem}{Theorem}
\newtheorem{property}{Property}
\newtheorem{corollary}{Corollary}

\newcommand{\newe}{e}
\newcommand{\newel}{e^{\ell}}
\newcommand{\newr}{r}
\newcommand{\newrbar}{\overline{r}}

\let\doendproof\endproof
\renewcommand\endproof{~\hfill\qed\doendproof}

\begin{document} 

\title{Composite repetition-aware data structures}
%\author{Djamal Belazzougui\inst{1,2} \fnmsep \thanks{Corresponding author: \email{name.surname@cs.helsinki.fi}} \and 
%Fabio Cunial\inst{1,2} \and
%Travis Gagie\inst{1,2} \fnmsep \thanks{Supported by the Academy of Finland.} \and \\
%Nicola Prezza\inst{3} \and
%Mathieu Raffinot\inst{4}}
%\institute{Department of Computer Science, University of Helsinki, Finland\thanks{This work was partially supported by Academy of Finland under grant 250345 (Center of Excellence in Cancer Genetics Research).}. \and
%Helsinki Institute for Information Technology, Finland. \and
%Department of Mathematics and Computer Science, University of Udine, Italy. \and
%LIAFA, Paris Diderot University - Paris 7, France.}

\author[1,2]{Djamal Belazzougui}
\author[1,2]{Fabio Cunial}
\author[1,2]{Travis Gagie}
\author[3]{Nicola Prezza}
\author[4]{Mathieu Raffinot}
\affil[1]{Department of Computer Science, University of Helsinki, Finland.\thanks{This work was partially supported by Academy of Finland under grant 250345 (Center of Excellence in Cancer Genetics Research).}}
\affil[2]{Helsinki Institute for Information Technology, Finland.}
\affil[3]{Department of Mathematics and Computer Science, University of Udine, Italy.}
\affil[4]{LIAFA, Paris Diderot University - Paris 7, France.}
\maketitle

\begin{abstract}
In highly repetitive strings, like collections of genomes from the
same species, distinct measures of repetition all grow sublinearly in
the length of the text, and indexes targeted to such strings typically
depend only on one of these measures. We describe two data structures whose size
depends on multiple measures of repetition at once, and that provide
competitive tradeoffs between the time for counting and reporting all the exact
occurrences of a pattern, and the space taken by the structure. The key component of our constructions is
the run-length encoded BWT (RLBWT), which takes space proportional to
the number of BWT runs: rather than augmenting RLBWT with suffix array
samples, we combine it with data structures from LZ77 indexes, which
take space proportional to the number of LZ77 factors, and with the
compact directed acyclic word graph (CDAWG), which takes space
proportional to the number of extensions of maximal repeats. The
combination of CDAWG and RLBWT enables also a new representation of
the suffix tree, whose size depends again on the number of extensions
of maximal repeats, and that is powerful enough to support matching
statistics and constant-space traversal.

%: which induces good time
%complexities on some aspect, for instance compressibility, but worse
%on others, on locating occurrences for instance

%The run-length compressed suffix array (RLCSA) is a powerful tool for indexing genomic databases and other repetitive datasets.  Its main drawback is its slow locating functionality, based on a sampled suffix array (SSA).  This is a problem even for standard FM-indexes but, since sets of similar genomes are much more compressible than single genomes, it is particularly acute for the RLCSA: to support reasonably fast locating, we often need an SSA that dwarfs the rest of the RLCSA.  In this paper we describe how the SSA can be replaced by either components from an LZ77 index or a compacted directed acyclic word graph (CDAWG).  Our theoretical and experimental analyses show...
\end{abstract}

\section{Introduction}\label{sec:introduction}

The space taken by compressed data structures for highly-repetitive
strings is typically a function of a specific measure of repetition,
for example the number $z$ of factors in a Lempel-Ziv parsing
\cite{arroyuelo2012stronger,kreft2013compressing}, or the number $r$
of runs in a Burrows-Wheeler transform \cite{MakinenNSV10}. For many
such compressed data structures, computing all the occurrences of a
pattern in the indexed string is a bottleneck. In this paper we
explore the advantages of \emph{combining data structures that depend
on distinct measures of repetition}. Specifically, we describe a
data structure that takes approximately $O(z+r)$ words of space, and
that reports all the occurrences of a pattern of length $m$ in
$O(m(\log{\log{n}} + \log{z}) + \mathtt{pocc}\log^{\epsilon}{z} +
\mathtt{socc}\log{\log{n}})$ time, where $n$ is the length of the
string and $\mathtt{pocc}$ and $\mathtt{socc}$ are the number of
primary and of secondary occurrences, respectively (see Section
\ref{sec:stringIndexes} for definitions). This compares favorably to
the $O(m^{2}h+(m+\mathtt{occ})\log{z})$ reporting time of LZ77 indexes
\cite{kreft2013compressing}, where $h$ is the height of the parse
tree. It also compares favorably in space to solutions based on
run-length encoded BWT (RLBWT) and suffix array samples
\cite{MakinenNSV10}, which take $O(n/k + r)$ words of space to achieve
$O(m\log{\log{n}} + k \cdot \mathtt{occ}\log{\log{n}})$ reporting
time, where $k$ is a sampling rate. %Finally, our result compares ??? to the $O(z\log(n/z))$ words of space of solutions based on straight-line programs, that report occurrences in $O(m\log{m}+\mathtt{occ}\log{\log{n}})$ time \cite{}.

%We also introduce a new measure of the repetitiveness of a string, the
%numbers $\newe$ and $\newel$ of right and left extensions of maximal repeats, which are related to the
%number of arcs in the compact directed acyclic word-graph (CDAWG) and
%which are an upper bound on $r$ and $z$. We show a data structure whose
%size depends on $\newe+\newel$ and that reports all the $\mathtt{occ}$
%occurrences of a pattern of length $m$ in a string of length $n$ in
%$O(m\log{\log{n}} + \mathtt{occ})$ time. The main component of our
%constructions is the RLBWT, which we use to count the number of
%occurrences of a pattern, and which we combine with the CDAWG and with
%data structures from LZ indexes, rather than with suffix array
%samples, for reporting. Similar combinations have already appeared in
%the literature, but their space has been related to statistical
%compressibility rather than to the number of repetitions: for example,
%an FM-index has already been combined with an LZ78 self-index to
%achieve faster search or reporting
%\cite{arroyuelo2012stronger,ferragina2005indexing}, but the size of
%the resulting data structure depends on $k$-th order empirical
%entropy.
We also introduce a new measure of the repetitiveness of a string, the
number $\newe$ of right extensions of maximal repeats, which is related to the
number of arcs in the compact directed acyclic word-graph (CDAWG) and
which is an upper bound on $r$ and $z$. We show a data structure whose
size depends on $\newe$ and that reports all the $\mathtt{occ}$
occurrences of a pattern of length $m$ in a string of length $n$ in
$O(m\log{\log{n}} + \mathtt{occ})$ time. The main component of our
constructions is the RLBWT, which we use to count the number of
occurrences of a pattern, and which we combine with the CDAWG and with
data structures from LZ indexes, rather than with suffix array
samples, for reporting. Similar combinations have already appeared in
the literature, but their space has been related to statistical
compressibility rather than to the number of repetitions: for example,
an FM-index has already been combined with an LZ78 self-index to
achieve faster search or reporting
\cite{arroyuelo2012stronger,ferragina2005indexing}, but the size of
the resulting data structure depends on $k$-th order empirical
entropy.

Combining the RLBWT with the CDAWG enables also a new representation
of the suffix tree, which takes space proportional to $\newe+\newel$
(where $\newel$ is the number of left extensions of maximal repeats) and which
supports a number of operations in $O(\log{\log{n}})$ time. Among
other properties, this new representation allows computing the
matching statistics of a pattern of length $m$ in $O(m\log{\log{n}})$
time.
Our constructions are targeted to highly-repetitive strings, like
large databases of similar genomes, in which all the measures of
repetition on which our data structures depend grow sublinearly in the
size of the database (see Figure \ref{fig:experiments} for an
example).

\begin{figure}[t]
\begin{center}
\includegraphics[width = 1\textwidth]{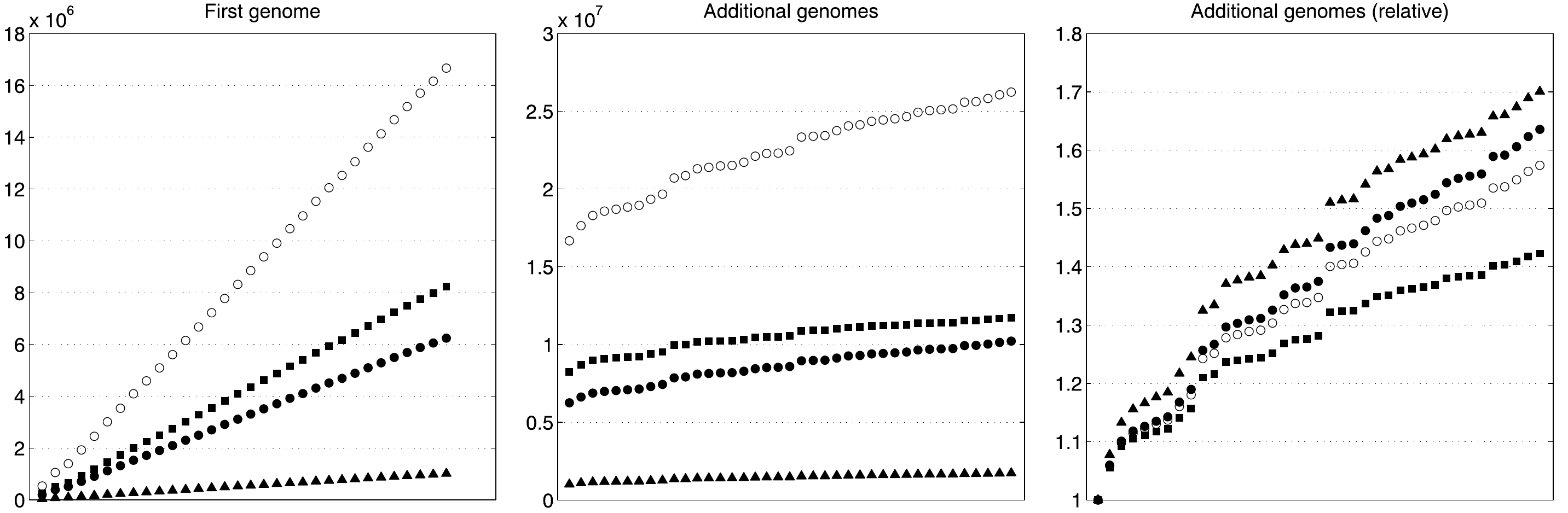}
\caption{Growth of the number of maximal repeats $|\mathcal{M}_T|$ (black circles), of $|\mathcal{E}^{r}_T \cup \mathcal{F}^{r}_T|$ (white circles, $e$ in the introduction), of the number of runs in BWT $|\mathcal{R}_T|$ (squares, $r$ in the introduction), and of $|\mathcal{Z}_T|$ (triangles, $z$ in the introduction) in a concatenation $T$ of 39 highly similar \emph{Saccharomyces cerevisiae} genomes \cite{pizzachili} (see Section \ref{sec:preliminaries} for definitions). Left: growth inside the first genome of the database. Center: growth after the addition of each genome (one sample per genome). Right: the same as the plot in the center, but with each curve normalized by its first sample. $|\mathcal{E}^{\ell}_T \cup \mathcal{F}^{\ell}_T|$, $|\mathcal{R}_{\REV{T}}|$ and $|\mathcal{Z}_{\REV{T}}|$ are not shown since they behave approximately as their symmetrical counterparts.
\label{fig:experiments}
}
\end{center}
\vspace{-1cm}
\end{figure}

\section{Preliminaries}\label{sec:preliminaries}

Let $\Sigma=[1..\sigma]$ be an integer alphabet, let $\#=0 \notin \Sigma$ be a separator, and let $T=[1..\sigma]^{n-1}\#$ be a string. We denote the reverse of $T$ by $\REV{T}$. Given a substring $W$ of $T$, let $\mathcal{P}_{T}(W)$ be the set of all starting positions of $W$ in the circular version of $T$. A \emph{repeat} $W$ is a string that satisfies $|\mathcal{P}_{T}(W)|>1$. We denote by $\Sigma^{\ell}_{T}(W)$ the set of characters $\{a \in [0..\sigma] : |\mathcal{P}_{T}(aW)|>0\}$ and by $\Sigma^{r}_{T}(W)$ the set of characters $\{b \in [0..\sigma] : |\mathcal{P}_{T}(Wb)|>0\}$. A repeat $W$ is \emph{right-maximal} (respectively, \emph{left-maximal}) iff $|\Sigma^{\ell}_{T}(W)|>1$ (respectively, iff $|\Sigma^{r}_{T}(W)|>1$). It is well known that $T$ can have at most $n-1$ right-maximal substrings and at most $n-1$ left-maximal substrings. A \emph{maximal repeat} of $T$ is a repeat that is both left- and right-maximal: we call $\mathcal{M}_{T}$ the set of all maximal repeats of $T$. A maximal repeat $W$ can be seen as a set of right-maximal substrings of $T$, and specifically as the set of all right-maximal strings $W[i..|W|]$ for $i \in [1..k]$ that are not left-maximal, and such that $W[k+1..|W|]$ is left-maximal.

For reasons of space we assume the reader to be familiar with the notion of \emph{suffix tree} $\ST_T = (V,E)$ of $T$, which we do not define here. We denote by $\ell(\gamma)$, or equivalently by $\ell(u,v)$, the label of edge $\gamma=(u,v) \in E$, and we denote by $\ell(v)$ the string label of node $v \in V$. It is well known that a substring $W$ of $T$ is right-maximal (respectively, left-maximal) iff $W=\ell(v)$ for some internal node $v$ of $\ST_T$ (respectively, iff $W=\REV{\ell(v)}$ for some internal node $v$ of $\ST_{\REV{T}}$). We assume the reader to be familiar with the notion of \emph{suffix link} connecting a node $v$ with $\ell(v)=aW$ for some $a \in [0..\sigma]$ to a node $w$ with $\ell(w)=W$: we say that $w=\mathtt{suffixLink}(v)$ in this case. Here we just recall that inverting the direction of all suffix links yields the so-called \emph{explicit Weiner links}. Given an internal node $v$ and a symbol $a \in [0..\sigma]$, it might happen that string $a\ell(v)$ does occur in $T$, but that it is not right-maximal, i.e. it is not the label of any internal node: all such left extensions of internal nodes that end in the middle of an edge are called \emph{implicit Weiner links}. An internal node can have more than one outgoing Weiner link, and all such Weiner links have distinct labels. %The number of suffix links (or, equivalently, of explicit Weiner links) is upper-bounded by $2n-2$, and the same bound holds for the number of implicit Weiner links.

The \emph{compact directed acyclic word graph} of a string $T$ (denoted by $\CDAWG_T$ in what follows) is the minimal compact automaton representing the set of suffixes of a given string \cite{blumer1987complete,CrochemoreV97}. It can be seen as the minimization of $\ST_T$, in which all leaves are merged to the same node (the sink) that represents $T$ itself, and in which all nodes except the sink are in one-to-one correspondence with the maximal repeats of $T$ \cite{Raffinot2001}. Since a maximal repeat corresponds to a set of right-maximal substrings, $\CDAWG_T$ can be built by putting in the
same equivalence class all nodes of $\ST_T$ that belong to the same maximal unary path of explicit Weiner links.
%before mathieu>The CDAWG of $T$, denoted by $\CDAWG_T$ in what follows, is a directed acyclic graph whose nodes are in one-to-one correspondence with the maximal repeats of $T$. Since a maximal repeat corresponds to a set of right-maximal substrings, $\CDAWG_T$ can be built by putting in the same equivalence class all nodes of $\ST_T$ that belong to the same maximal unary path of explicit Weiner links. All leaves of $\ST_T$ are mapped to the same artificial node of $\CDAWG_T$ (the sink), which represents $T$ itself rather than a maximal repeat.

%The \emph{suffix array} $\SA_{T}[1..n]$ of $T$ is the permutation of $[1 .. n]$ such that $\SA_{T}[p]=i$ iff suffix $T[i..n]$ has position $p$ in the list of all suffixes of $T$ taken in lexicographic order. The Burrows-Wheeler transform of $T$ is the string $\BWT_{T}[1..n]$ that satisfies $\BWT_{T}[p]=T[\SA_{T}[p]-1]$ if $\SA_{T}[p]>0$, and $\BWT_{T}[p]=\#$ otherwise. 
For reasons of space we assume the reader to be familiar with the notion and uses of the Burrows-Wheeler transform of $T$, including the $C$ array and backward searching. In this paper we use $\BWT_T$ to denote the BWT of $T$, and we use $\INTERVAL{W} = [\SP{W}..\EP{W}]$ to denote the lexicographic interval of a string $W$ in a BWT that is implicit from the context. We say that $\BWT_{T}[i..j]$ is a \emph{run} iff $\BWT_{T}[k]=c \in [0..\sigma]$ for all $k \in [i..j]$, and moreover if any substring $\BWT_{T}[i'..j']$ such that $i' \leq i$, $j' \geq j$, and either $i' \neq i$ or $j' \neq j$, contains at least two distinct characters. It is well known that repetitions in $T$ tend to be converted into runs of $\BWT_T$. %: for example, the BWT of $k$ copies of a string of length $n$ is a string consisting of $k$ runs each of length $n$. 
We denote by $\mathcal{R}_{T}$ the set of all triplets $(c,i,j)$ such that $\BWT_{T}[i..j]$ is a run of character $c$, and we use $\newr_T$ and $\newrbar_T$ as shorthands for $|\mathcal{R}_T|$ and $|\mathcal{R}_{\REV{T}}|$, respectively.

The \emph{LZ77 factorization} of $T$ \cite{ziv1977universal} is the greedy decomposition $T_1 T_2 \cdots T_{z}$ of $T$ obtained as follows. Assume that $T$ is virtually preceded by the $\sigma$ distinct characters in its alphabet, and assume that $T_1 T_2 \cdots T_i$ has already been computed for some prefix of length $k$ of $T$: then, $T_{i+1}$ is the longest prefix of $T[k+1..n]$ such that there is a $j \leq k$ that satisfies $T[j..j+|T_{i+1}|-1] = T_{i+1}$. We denote by $\mathcal{Z}_T$ the set of pairs $(T_i,p_i)$ for all $i \in [1..z]$, where $p_i$ is the starting position of $T_i$ in $T$, and we use $z_T$ as a shorthand for $|\mathcal{Z}_T|$. From now on, we drop subscripts whenever the string $T$ they specify is clear from the context.

\subsection{Relationships among maximal repeats, runs in BWT, and LZ factors}

Clearly $|\mathcal{R}|$ can be as small as two, e.g. in string $\mathtt{0}^{n-1}\#$, and as large as $\Theta(n)$, e.g. in the string of length $n$ that contains exactly $n$ distinct characters, or in a de Bruijn string of order $k>1$ on a binary alphabet: this string of length $\sigma^k+k-1$ contains all the distinct $k$-mers, thus the interval of every $(k-1)$-mer in $\BWT_T$ contains exactly $\sigma$ distinct characters, and the number of runs in $\BWT_T$ is thus at least $\sigma^{k-1}(k-1)$. It is known that $|\mathcal{Z}|$ is $O(n/\log_{\sigma}n)$ \cite{lempel1976complexity}, and it can be constant, e.g. in $\mathtt{0}^{n-1}\#$. Conversely, $|\mathcal{M}|$ can be zero, e.g. in a string of length $n$ that contains exactly $n$ distinct characters, and it can be $\Theta(n)$ in the worst case, e.g. in string $\mathtt{0}^{n-1}\#$. When maximal repeats exist, the number of \emph{right extensions of maximal repeats} $\sum_{W \in \mathcal{M}}|\Sigma^{r}(W)|$ is $\Omega(\log{n})$ (see Lemma \ref{lemma:nArcsInCDAWG} in the appendix), and this lower bound is matched by Fibonacci strings and by Thue-Morse strings of length $n$, whose CDAWG contains $O(\log{n})$ nodes \cite{Radoszewski2012,Rytter06}. Both $|\mathcal{M}|/|\mathcal{R}|$ and $|\mathcal{M}|/|\mathcal{Z}|$ can be $\Theta(n)$, for example in the already mentioned $\mathtt{0}^{n-1}\#$. $|\mathcal{R}|/|\mathcal{Z}|$ can be $\Theta(\log n)$, e.g. in the already mentioned de Bruijn string $T$ of order $k$, which has $\Theta(n/\log_{\sigma}{n})$ LZ factors. However, $|\mathcal{M}|$, $|\mathcal{R}|$ and $|\mathcal{Z}|$ can all grow at the same asymptotic rate in the same family of strings. Consider e.g. string $T = \mathtt{0}^{1}\mathtt{1}\mathtt{0}^{2}\mathtt{1} \cdots \mathtt{0}^{x}\mathtt{1}\#$ of length $x(x+3)/2+1$. Clearly $|\mathcal{Z}|=x+3$, and $|\mathcal{M}|=3(x-1)$ since the maximal repeats of $T$ are only the substrings $\mathtt{0}^{i}\mathtt{1}$ for $i \in [1..x-1]$, $\mathtt{0}^{j}$ for $j \in [1..x-1]$, and $\mathtt{0}^{k-1}\mathtt{1}\mathtt{0}^{k}$ for $k \in [2..x-1]$. Replacing $\#$ with a new block $\mathtt{0}^{x+1}\mathtt{1}\#$ in string $T$ creates two new runs for every $x>1$%\footnote{Specifically, the two new runs are created by suffix $\mathtt{0}^{x}\#$.}
, thus $|\mathcal{R}|=2x$ for $x>1$.

Recall that a substring $W$ of $T$ is a maximal repeat iff $W=\ell(v)$ for some internal node $v$ of $\ST_T = (V,E)$, and moreover if there are at least two Weiner links from $v$. Since the set of all left-maximal substrings of $T$ is closed under the prefix operation, there is a bijection between $\mathcal{M}$ and the nodes that lie on the paths of $\ST_T$ that start from the root and that end at nodes labeled by maximal repeats defined as follows:

\begin{definition}
A maximal repeat $W$ of a string $T \in [1..\sigma]^{n-1}\#$ is \emph{rightmost} if no string $WV$ with $V \in [0..\sigma]^+$ is left-maximal in $T$.
\end{definition}

We denote the set of rightmost maximal repeats of $T$ by $\mathcal{M}^{r}_{T}$. We also denote by $\mathcal{E}^{r}_T$ the set of edges of $\ST_T$ that connect pairs of nodes labeled by maximal repeats, and we denote by $\mathcal{F}^{r}_T$ the set of edges $(v,w)$ in $\ST_T$ such that $\ell(v) \in \mathcal{M}_T$ and $\ell(w) \notin \mathcal{M}_T$. We use $\mathcal{M}^{\ell}_{T}$, $\mathcal{E}^{\ell}_T$ and $\mathcal{F}^{\ell}_T$ to denote symmetrical concepts in $\ST_{\REV{T}}$, and we use $\newe_T$ and $\newel_T$ as shorthands for $|\mathcal{E}^{r}_T| + |\mathcal{F}^{r}_T|$ and for $|\mathcal{E}^{\ell}_T| + |\mathcal{F}^{\ell}_T|$, respectively. 
%We can now formally define the unformal $e$ of the introduction as $e = |\mathcal{E}^{r}_T| + |\mathcal{F}^{r}_T|$. 
%, and we define $e^{\ell} = |\mathcal{E}^{\ell}_T| + |\mathcal{F}^{\ell}_T|$. 
Clearly $\mathcal{E}^{r}$ and $\mathcal{F}^{r}$ are the image of explicit and implicit Weiner links of $\ST_{\REV{T}}$:

\begin{lemma}\label{lemma:correspondence}
Let $\ST_T=(V,E)$. There is a bijection between $\mathcal{E}^{r}_T$ and the set of all explicit Weiner links from nodes of $\ST_{\REV{T}}$ that correspond to maximal repeats of $T$. There is a bijection between $\mathcal{F}^{r}_T$ and the set of all implicit Weiner links from nodes of $\ST_{\REV{T}}$ that correspond to maximal repeats of $T$.
\end{lemma}
%\begin{proof}
%Strings in $\mathcal{M}^{\ell}$, reversed, label internal nodes of $\ST_{\REV{T}}$ from which no edge in $\mathcal{E}^{\ell}$ starts, thus, by Lemma \ref{lemma:correspondence}, the reverse of such strings label internal nodes of $\ST_T$ from which no explicit Weiner link starts.
%\end{proof}

The proof of Lemma \ref{lemma:correspondence} is provided in the appendix. It is clear that the set of suffix tree edges $\mathcal{E}^{r}_T \cup \mathcal{F}^{r}_T$ is in one-to-one correspondence with the set of all arcs of $\CDAWG_T$. This set of edges is also related to runs in $\BWT_T$:

\begin{theorem}\label{thm:numberOfRuns}
$| [0..\sigma] \setminus \cup_{W \in \mathcal{M}^{r}_T} \Sigma^{\ell}_{T}(W) | + \sum_{W \in \mathcal{M}^{r}_T}|\Sigma^{\ell}_{T}(W)|-|\mathcal{M}^{r}_T|+1 \leq |\mathcal{R}_T| \leq |\mathcal{F}^{r}_T|$.
\end{theorem}
\begin{proof}
The root of $\ST_T$ is a maximal repeat, thus the destinations of all edges in $\mathcal{F}^{r}$ partition all leaves of $\ST_T$ into disjoint subtrees, or equivalently they partition the entire $\BWT_T$ in disjoint blocks. Since every such block is the interval in $\BWT_T$ of some string that is not left-maximal, all characters of $\BWT_T$ in the same block are identical, thus the number of runs in $\BWT_T$ cannot be bigger than $|\mathcal{F}^{r}|$.

The interval of a string $W \in \mathcal{M}^r$ in $\BWT_T$ contains exactly $|\Sigma^{\ell}(W)|$ distinct characters, and at most one of them is identical to the character that precedes the largest suffix of $T$ smaller than $W$ in lexicographic order (note that such suffix might not be prefixed by any string in $\mathcal{M}^r$). Thus, the number of runs in $\BWT_T$ is at least $\sum_{W \in \mathcal{M}^{r}}|\Sigma^{\ell}(W)|-|\mathcal{M}^{r}|+1$. Factor $[0..\sigma] \setminus \cup_{W \in \mathcal{M}^{r}} \Sigma^{\ell}_{T}(W)$ in the claim takes into account symbols of $T$ that never occur to the left of strings in $\mathcal{M}^r$.
\end{proof}

A symmetrical argument holds for $\mathcal{R}_{\REV{T}}$. The set of arcs in $\CDAWG_T$ is also related to the LZ factorization of $T$:

%% \begin{theorem}
%% $z \leq |\mathcal{E}^r_T \cup \mathcal{F}^r_T|$
%% \end{theorem}

\begin{theorem}
$|\mathcal{Z}_T| \leq |\mathcal{E}^r_T \cup \mathcal{F}^r_T|$
%$z \leq \newe$
\end{theorem}
\begin{proof}
Let $T=T_1 T_2 \dots T_z$ be the LZ factorization of $T$, and let $p_1,p_2,\dots,p_z$ be the sequence such that $p_i$ is the starting position of factor $T_i$ in $T$. Every factor is a right-maximal substring of $T$, but it is not necessarily left-maximal: let $W_i$ be a suffix of $T[1..p_i-1]$ such that $W_{i}T_i$ is both right-maximal and left-maximal, and assume that we assign $T_i$ to the edge $(v,w)$ in $\mathcal{E}^r_T \cup \mathcal{F}^r_T$ such that $\ell(v)=W_{i}T_i$, $v=\mathtt{parent}(w)$, and the first character of $T_{i+1}$ equals the first character of $\ell(v,w)$. Assume that there is some $j>i$ for which we assign $T_j$ to the same maximal repeat $W_{i}T_i$. Then, the first character of $T_{j+1}$ must be different from the first character of $T_{i+1}$, otherwise factor $T_j$ would have been longer. It follows that every LZ factor can be assigned to a distinct element of $\mathcal{E}^r_T \cup \mathcal{F}^r_T$.
\end{proof}

%% The gap between $|\mathcal{R}_T|$ and $|\mathcal{E}^{r}_T \cup \mathcal{F}^{r}_T|$, and between $z$ and $|\mathcal{E}^r_T \cup \mathcal{F}^r_T|$, is apparent from Figure \ref{fig:experiments} (center). However, all these measures seem to grow at the same relative rate in practice (right panel).

The gap between $\newr$ %$|\mathcal{R}_T|$ 
and $e$, and between $z$ and $e$, is apparent from Figure \ref{fig:experiments} (center). However, all these measures seem to grow at the same relative rate in practice (right panel).

\subsection{Repetition-aware data structures}\label{sec:stringIndexes}

Given a string $T \in [1..\sigma]^{n-1}\#$, we call \emph{run-length encoded BWT} any representation of $\BWT_T$ that takes $O(|\mathcal{R}_T|)$ words of space, and that supports rank and select operations: see for example \cite{makinen2005succinct1,MakinenNSV10,SirenVMN08}. Let $\mathcal{R}_T$ be a set of triplets $(c,i,j)$ such that $\BWT_{T}[i..j]$ is a run of character $c$. It is easy to implement rank in $O(\log{\log{n}})$ time, by encoding $\mathcal{R}_T$ as $\sigma+1$ predecessor data structures \cite{Wi83}, each of which stores the second component of all triplets with the same first component. For every such second component $i$, we also store in an array the sum of all occurrences of $c$ up to $i$, exclusive. To implement select in $O(\log{\log{n}})$ time, we can similarly encode $\mathcal{R}_T$ as $\sigma+1$ predecessor data structures, each of which stores value $\mathtt{rank}_{c}(\BWT_T,i-1)$ for all triplets $(c,i,j)$ with the same value of $c$. We also store the value of $i$ for every such triplet. We denote the run-length encoded BWT of $T$ by $\RLCSA_T$.
%
%The RLCSA consists of a run-length compressed rank data structure for the BWT of the dataset, and a sampled suffix array (SSA).  Given a pattern, we use the rank data structure to find the interval in the BWT containing characters preceding occurrences of that pattern in the dataset; the length of this interval, uncompressed, is the number of those occurrences.  If we want to locate a particular occurrence in the dataset, we start at the preceding character and use rank queries to move backward until we reach a character whose position we have sampled.  Thus, the average time for locating is inversely proportional to the size of the SSA.  Therefore, to support reasonably fast locating, we need a fairly large SSA, regardless of the compressibility of the dataset.

For reasons of space we assume the reader to be familiar with LZ77-indexes: see e.g. \cite{karkkainen1996lempel,gagie2014lz77}. Here we just recall that a \emph{primary occurrence} of a pattern $P$ in a string $T \in [1..\sigma]^{n-1}\#$ is one that crosses or ends at a phrase boundary in the LZ77 factorization $T_1 T_2 \cdots T_{z}$ of $T$. All other occurrences are called \emph{secondary}. Once we have determined all primary occurrences, locating secondary occurrences reduces to two-sided range reporting and takes $O(\occ \log{\log{n}})$ time with a data structure that takes $O(z)$ words of space \cite{karkkainen1996lempel}. To locate primary occurrences, we can use a data structure for four-sided range reporting on a \(z \times z\) grid, with a marker at \((x, y)\) if the $x$th LZ factor in lexicographic order is preceded in the text by the lexicographically $y$th reversed prefix ending at a phrase boundary. This data structure takes $O(z)$ words of space, and it returns all the phrase boundaries immediately followed by a factor in the specified range, and immediately preceded by a reversed prefix in the specified range, in $O((1+k)\log^{\epsilon}{z})$ time, where $k$ is the number of phrase boundaries reported \cite{chan2011orthogonal}.

\section{Combining runs in BWT and LZ factors}\label{sec:lz77}

In this section we describe how to combine data structures whose size depends on the number of LZ factors of a string $T \in [1..\sigma]^{n-1}\#$, and data structures whose size depends on the number of runs in $\BWT_T$, to report all the occurrences of a pattern in $T$. To do so, we first need to solve the following subproblem. Let $\ST_{T}=(V,E)$ be the suffix tree of $T$, and let $V'=\{v_1,v_2,\dots,v_k\} \subseteq V$ be a subset of the nodes of $\ST_T$. Consider the list of node labels $L = \ell(v_1),\ell(v_2),\dots,\ell(v_k)$, sorted in lexicographic order. Given a string $W \in [0..\sigma]^*$, we want to implement function $\INTERVALFUNCTION(W,V')$ that returns the (possibly empty) interval of $W$ in $L$. The following lemma describes how to do this in $O(k)$ words of space:

\begin{lemma}\label{lemma:intervalOperation2}
Let $T \in [1..\sigma]^{n-1}\#$ be a string, and let $V'$ be a subset of $k$ nodes of its suffix tree, represented as intervals in $\BWT_T$. Given the interval $[i..j]$ of a string $W \in [0..\sigma]^*$ in $\BWT_T$, there is a data structure that takes $O(k)$ words of space and that computes $\INTERVALFUNCTION(W,V')$ in $O(\log{k})$ time.
\end{lemma}
\begin{proof}
We store a bitvector $\mathtt{first}[1..n]$ such that $\mathtt{first}[i]=1$ iff there is a node $v' \in V'$ such that $\INTERVAL{v'}=[i..j]$. Similarly, we store a bitvector $\mathtt{last}[1..n]$ such that $\mathtt{last}[j]=1$ iff there is a node $v' \in V'$ such that $\INTERVAL{v'}=[i..j]$. Let $\alpha$ and $\beta$ be the number of ones in $\mathtt{first}$ and $\mathtt{last}$, respectively. We build prefix-sum arrays $\mathtt{First}$ and $\mathtt{Last}$ on such bitvectors using $O(k)$ words of space, and we discard $\mathtt{first}$ and $\mathtt{last}$. Let $F[1..\alpha]$ be the array such that $F[i]$ equals the number of intervals $[p..q]$ such that $p$ is the $i$th one in $\mathtt{first}$ and $[p..q]=\INTERVAL{v'}$ for a node $v' \in V'$. Similarly, let $L[1..\beta]$ be the array such that $L[i]$ equals the number of intervals $[p..q]$ such that $q$ is the $i$th one in $\mathtt{last}$ and $[p..q]=\INTERVAL{v'}$ for a node $v' \in V'$. %It is easy to build $F$ and $L$ by sorting the list of intervals in $O(n+k)$ time and space using radix sort. 
We represent $F$ and $L$ as prefix-sum arrays using $O(k)$ words of space, and we discard $F$ and $L$.

Let $\INTERVALFUNCTION(W,V')=[x..y]$. Given the interval $[i..j]$ of a string $W$ in $\BWT_T$, we find the corresponding interval $[i'..j']$ in array $\mathtt{first}$ in $O(\log \alpha)$ time, using binary search on $\mathtt{First}$. Specifically, $i'=\min\{p \in [1..\alpha] : \mathtt{First}[p] \geq i\}$ and $j'=\max\{q \in [1..\alpha] : \mathtt{First}[q] \leq j\}$. If $j'<i'$ then $W$ is not the prefix of a label of a node in $V'$. Otherwise, since all nodes $v' \in V'$ whose BWT interval starts inside $[i+1..j]$ are right extensions of $W$, we set $y=\sum_{p=1}^{j'}F[p]$ in constant time using the prefix-sum representation of $F$. If $\mathtt{First}[i'] \neq i$, i.e. if no interval of a node $v' \in V'$ starts at position $i$ in $\BWT_T$, then we can just set $x=1+\sum_{p=1}^{i'-1}F[p]$ and stop.

Otherwise, it could happen that just a (possibly empty) subset of all the nodes in $V'$ whose interval starts at position $i$ in $\BWT_T$ correspond to $W$ or to right extensions of $W$: the intervals of such nodes necessarily end inside $[i..j]$. All the other intervals that start at position $i$ could correspond instead to \emph{prefixes} of $W$, and they necessarily end after position $j$ in $\BWT_T$. Thus, let $[i''..j'']$ be the interval in $\mathtt{last}$ that corresponds to $[i..j]$: specifically, let $i''=\min\{p \in [1..\beta] : \mathtt{Last}[p] \geq i\}$ and $j''=\max\{q \in [1..\beta] : \mathtt{Last}[q] \leq j\}$. To determine the number of intervals that start at position $i$ in $\BWT_T$ and that correspond to prefixes of $W$, it suffices to compute the difference $\delta$ between the number of starting positions and the number of ending positions inside interval $[i..j]$, as follows: $\delta = \left( \sum_{p=1}^{j'}F[p] - \sum_{p=1}^{i'-1}F[p] \right) - \left( \sum_{q=1}^{j''}L[q] - \sum_{q=1}^{i''-1}L[q] \right)$. Then, $x = \sum_{p=1}^{i'}F[p] - \delta$. All such sums can be computed in constant time using the prefix-sum representations of $F$ ad $L$.
\end{proof}
\vspace{-0.4cm}

Consider now a factorization of $T$ such that all factors are right-maximal substrings of $T$, and let $V'$ be the set of nodes of $\ST_T$ that correspond to the distinct factors. To locate all the occurrences of a pattern that cross or end at a boundary between two factors, we just need an implementation of function $\INTERVALFUNCTION(W,V')$ and a pair of RLBWTs:

%% \begin{lemma}\label{lemma:straddlingOccurrences}
%% Let $T \in [1..\sigma]^{n-1}\#$ be a string, and let $T = T_1 T_2 \cdots T_z$ be a factorization of $T$ in which all factors are right-maximal substrings. There is a data structure that takes $O(z + |\mathcal{R}_T| + |\mathcal{R}_{\REV{T}}|)$ words of space and that reports all the $\mathtt{occ}$ occurrences of a pattern $P \in [0..\sigma]^m$ that cross or end at a boundary between two factors of $T$, in $O(m(\log\log n + \log z) + \mathtt{occ}\log^{\epsilon}{z})$ time.
%% \end{lemma}

\begin{lemma}\label{lemma:straddlingOccurrences}
Let $T \in [1..\sigma]^{n-1}\#$ be a string, and let $T = T_1 T_2 \cdots T_z$ be a factorization of $T$ in which all factors are right-maximal substrings. There is a data structure that takes $O(z + \newr_T+\newrbar_T)$ %$O(z + r)$ 
words of space and that reports all the $\mathtt{occ}$ occurrences of a pattern $P \in [0..\sigma]^m$ that cross or end at a boundary between two factors of $T$, in $O(m(\log\log n + \log z) + \mathtt{occ}\log^{\epsilon}{z})$ time.
\end{lemma}
\begin{proof}
Let $p_1,p_2,\dots,p_z$ be the sequence such that $p_i$ is the starting position of factor $T_i$ in $T$. The same occurrence of $P$ in $T$ can cover up to $m$ boundaries between two factors, thus we organize the computation as follows. We consider every possible way to place \emph{the rightmost boundary between two factors} in $P$, i.e. every possible split of $P$ into two parts $P[1..k-1]$ and $P[k..m]$ for $k \in [1..m]$, such that $P[k..m]$ is either a factor or a proper prefix of a factor. For every such $k$, we use four-sided range reporting queries to list all the occurrences of $P$ in $T$ that conform to this split, as described in Section \ref{sec:stringIndexes}. The four-sided range reporting data structure represents the mapping between the lexicographic rank of a factor $W$ among all the distinct factors of $T$, and the lexicographic ranks of all the reversed prefixes $\REV{T[1..p_i-1]}$ such that $T_i=W$, among all the reversed prefixes of $T$ that end at the last position of a factor. As described in Section \ref{sec:stringIndexes}, this data structure takes $O(z)$ words of space.

We encode sequence $p_1,p_2,\dots,p_z$ implicitly, as follows: we use a bitvector $\mathtt{last}[1..n]$ such that $\mathtt{last}[i]=1$ iff $\SA_{\REV{T}}[i]=n-p_j+2$ for some $j \in [1..z]$, i.e. iff $\SA_{\REV{T}}[i]$ is the last position of a factor. We represent such bitvector as a predecessor data structure with partial ranks, using $O(z)$ words of space \cite{Wi83}. Then, we build the data structure described in Lemma \ref{lemma:intervalOperation2}, where $V'$ is the set of loci in $\ST_{T}$ of all factors of $T$. This data structure takes $O(z)$ words of space, and together with $\mathtt{last}$, $\RLCSA_T$ and $\RLCSA_{\REV{T}}$, it is the output of our construction. 

Given a pattern $P \in [0..\sigma]^m$, we first perform a backward search in $\RLCSA_T$ to determine the number of occurrences of $P$ in $T$: if this number is zero, we stop. During this backward search, we store in a table the interval $[i_k .. j_k]$ of $P[k..m]$ in $\BWT_T$ for every $k \in [2..m]$. Then, we compute the interval $[i'_{k-1} .. j'_{k-1}]$ of $\REV{P[1..k-1]}$ in $\BWT_{\REV{T}}$ for every $k \in [2..m]$, using backward search in $\RLCSA_{\REV{T}}$: if $\mathtt{rank}_{1}(\mathtt{last},j'_{k-1}) - \mathtt{rank}_{1}(\mathtt{last},i'_{k-1}-1) = 0$, then $P[1..k-1]$ never ends at the last position of a factor, and we can discard this value of $k$. Otherwise, we convert $[i'_{k-1} .. j'_{k-1}]$ to the interval $[\mathtt{rank}_{1}(\mathtt{last},i'_{k-1})+1 .. \mathtt{rank}_{1}(\mathtt{last},j'_{k-1})]$ of all the reversed prefixes of $T$ that end at the last position of a factor. Rank operations on $\mathtt{last}$ can be implemented in $O(\log{\log{n}})$ time using predecessor queries. We get the lexicographic interval of $P[k..m]$ in the list of all the distinct factors of $T$ using operation $\INTERVALFUNCTION(P[k..m],V')$, in $O(\log z)$ time. We use such intervals to query the four-sided range reporting data structure.
\end{proof}

The algorithm described in Lemma \ref{lemma:straddlingOccurrences} can be engineered in a number of ways in practice. Here we just apply it to the LZ factorization of $T$ to find all the primary occurrences of $P$ in $T$, and we use the strategy described in Section \ref{sec:stringIndexes} to compute secondary occurrences, obtaining the key result of this section:

%% \begin{theorem}
%% Let $T \in [1..\sigma]^{n-1}\#$ be a string, and let $T = T_1 T_2 \dots T_z$ be its LZ factorization. There is a data structure that takes $O(z + |\mathcal{R}_T| + |\mathcal{R}_{\REV{T}}|)$ words of space and that reports all the $\mathtt{pocc}$ primary occurrences and all the $\mathtt{sOcc}$ secondary occurrences of a pattern $P \in [0..\sigma]^m$ in $O(m(\log\log n + \log z) + \mathtt{pocc}\log^{\epsilon}{z} + \mathtt{socc}\log{\log{n}})$ time.
%% \end{theorem}

\begin{theorem}
Let $T \in [1..\sigma]^{n-1}\#$ be a string, and let $T = T_1 T_2 \dots T_z$ be its LZ factorization. There is a data structure that takes $O(z+\newr_T+\newrbar_T)$ %$O(z + r)$ 
words of space and that reports all the $\mathtt{pocc}$ primary occurrences and all the $\mathtt{socc}$ secondary occurrences of a pattern $P \in [0..\sigma]^m$ in $O(m(\log\log n + \log z) + \mathtt{pocc}\log^{\epsilon}{z} + \mathtt{socc}\log{\log{n}})$ time.
\end{theorem}

%Note that, thanks to $\RLCSA_T$, we know $\mathtt{pOcc}+\mathtt{sOcc}$ before starting to locate occurrences: we can thus stop locating occurrences in practice as soon as we have found all of them.

%Second, we avoid storing the Patricia trees, using a technique described in \cite{arroyuelo2012stronger}. Given a candidate partitioning of $P$ into $P[1..k]$ and $P[k+1..m]$, we can first use $\RLCSA_T$ to check whether the interval $[i..j]$ of $P[k+1..m]$ in $\SA_T$ is nonempty: if this is the case, and if $\mathtt{first}[i..j]$ contains at least a one, then the interval of $P[k+1..m]$ in $\mathcal{T}$ is $[\mathtt{rank}_{1}(\mathtt{first},i-1)+1 .. \mathtt{rank}_{1}(\mathtt{first},j) ]$. Similarly, we can use $\RLCSA_{\REV{T}}$ to check whether the interval $[i'..j']$ of $P[1..k]$ in $\SA_{\REV{T}}$ is nonempty: if this is the case, and if $\mathtt{last}[i'..j']$ contains at least a one, then the interval of $P[1..k]$ in $\REV{\mathcal{T}}$ is $[\mathtt{rank}_{1}(\mathtt{last},i'-1)+1 .. \mathtt{rank}_{1}(\mathtt{last},j') ]$. The following result is thus immediate:

\section{Combining runs in BWT and maximal repeats}\label{sec:cdawg}

An alternative way to compute all the occurrences of a pattern in a string $T$ consists in combining $\RLCSA_T$ with $\CDAWG_T$, using an amount of space proportional to the number of right extensions of the maximal repeats of $T$:

%% \begin{theorem}\label{thm:locate_dawg}
%% Let $T \in [1..\sigma]^{n-1}\#$ be a string. There is a data structure that takes $O(|\mathcal{E}^{r}_T \cup \mathcal{F}^{r}_T|)$ words of space (or alternatively, $O(|\mathcal{E}^{\ell}_T \cup \mathcal{F}^{\ell}_T|)$ words of space) and that reports all the $\mathtt{occ}$ occurrences of a pattern $P \in [0..\sigma]^m$ in $O(m\log{\log{n}} + \mathtt{occ})$ time.
%% \end{theorem}

\begin{theorem}\label{thm:locate_dawg}
Let $T \in [1..\sigma]^{n-1}\#$ be a string. There is a data structure that takes $O(\newe_T)$ %$O(\newe)$ 
words of space (or alternatively, $O(\newel_T)$ %$O(e^{\ell})$ 
words of space) and that reports all the $\mathtt{occ}$ occurrences of a pattern $P \in [0..\sigma]^m$ in $O(m\log{\log{n}} + \mathtt{occ})$ time.
\end{theorem}
\begin{proof}
We build $\RLCSA_T$ and $\CDAWG_T$. For every node $v$ in the CDAWG, we store $|\ell(v)|$ in a variable $v.\mathtt{length}$. Recall that an arc $(v,w)$ of the CDAWG means that maximal repeat $\ell(w)$ can be obtained by extending maximal repeat $\ell(v)$ to the right and to the left. Thus, for every arc $\gamma=(v,w)$ of $\CDAWG_T$, we store the first character of $\ell(\gamma)$ in a variable $\gamma.\mathtt{char}$, and we store the length of the right extension implied by $\gamma$ in a variable $\gamma.\mathtt{right}$. The length $\gamma.\mathtt{left}$ of the left extension implied by $\gamma$ can be computed by $w.\mathtt{length}-v.\mathtt{length}-\gamma.\mathtt{right}$. Clearly arcs of $\CDAWG_T$ that correspond to edges of $\ST_T$ in set $\mathcal{E}^{r}_T$ induce no left extension. For every arc of $\CDAWG_T$ that connects a maximal repeat $W$ to the sink, we store just $\gamma.\mathtt{char}$ and the starting position $\gamma.\mathtt{pos}$ of string $W \cdot \gamma.\mathtt{char}$ in $T$. The total space used by the CDAWG is clearly $O(\newe)$ words, and by Theorem \ref{thm:numberOfRuns} the space used by $\RLCSA_T$ is $O(|\mathcal{F}^{r}_T|)$ words. An alternative construction could use $\CDAWG_{\REV{T}}$ and $\RLCSA_{\REV{T}}$.

We use the RLBWT to count the number of occurrences of $P$ in $T$ in $O(m\log{\log{n}})$ time: if this number is greater than zero, we use the CDAWG to report all the $\mathtt{occ}$ occurrences of $P$ in $T$ in $O(\mathtt{occ})$ time, using the technique sketched in \cite{crochemore1997automata}. Specifically, since we know that $P$ occurs in $T$, we perform a blind search for $P$ in the CDAWG, as is typically done with Patricia trees. We keep a variable $i$, initialized to zero, that stores the length of the prefix of $P$ that we have matched so far, and we keep a variable $j$, initialized to one, that stores the starting position of $P$ inside the last maximal repeat encountered during the search. For every node $v$ in the CDAWG, we choose the arc $\gamma$ such that $\gamma.\mathtt{char}=P[i+1]$ in constant time using hashing, we increment $i$ by $\gamma.\mathtt{right}$, and we increment $j$ by $\gamma.\mathtt{left}$. If the search leads to the sink by an arc $\gamma$, we report $\gamma.\mathtt{pos}+j$ and we stop. If the search leads to a node $v$ that is associated with the maximal repeat $W$, we determine all the occurrences of $W$ in $T$ by performing a depth-first traversal of all the nodes in the CDAWG that are reachable from $v$, updating variables $i$ and $j$ as described above, and reporting $\gamma.\mathtt{pos}+j$ for every arc $\gamma$ that leads to the sink. The total number of nodes and arcs reachable from $v$ is clearly $O(\mathtt{occ})$.
\end{proof}

The combination of $\CDAWG_T$ and $\RLCSA_T$ can also be used to implement a repetition-aware representation of $\ST_T$. We will apply the following property to support operations on $\ST_T$:

\begin{property}\label{obs:interval}
A maximal repeat $W=[1..\sigma]^m$ of $T$ is the equivalence class of all the right-maximal strings $\{W[1..m],\dots,W[k..m]\}$ such that $W[k+1..m]$ is left-maximal, and $W[i..m]$ is not left-maximal for all $i \in [2..k]$. Equivalently, the node $v'$ of $\CDAWG_T$ with $\ell(v')=W$ is the equivalence class of the nodes $\{v_1,\dots,v_k\}$ of $\ST_T$ such that $\ell(v_i)=W[i..m]$ for all $i \in [1..k]$, and such that $v_k,v_{k-1},\dots,v_1$ is a maximal unary path of Weiner links.
\end{property}

Thus, the set of right-maximal strings that belong to the equivalence class of a maximal repeat can be represented by a single integer $k$, and a right-maximal string can be identified by the maximal repeat $W$ it belongs to, and by the length of the corresponding suffix of $W$. In $\BWT_T$, the right-maximal strings in the same equivalence class enjoy the following additional properties:

\begin{property}\label{obs:equivalenceClassInBWT}
Let $\{W[1..m],\dots,W[k..m]\}$ be the right-maximal strings that belong to the equivalence class of maximal repeat $W \in [1..\sigma]^m$, and let $\INTERVAL{W[i..m]}=[p_i..q_i]$ for $i \in [1..k]$. Then:
\begin{enumerate}
\item $|q_i-p_i+1|=|q_j-p_j+1|$ for all $i$ and $j$ in $[1..k]$.
\item $\BWT_{T}[p_i..q_i]=W[i-1]^{q_i-p_i+1}$ for $i \in [2..k]$. Conversely, $\BWT_{T}[p_1..q_1]$ contains at least two distinct characters.
\item $p_{i-1} = C[c]+\mathtt{rank}_{c}(\BWT_T,p_i)$ and $q_{i-1}=p_{i-1}+q_i-p_i$ for $i \in [2..k]$, where $c=W[i-1]=\BWT_{T}[p_i]$.\label{point:weinerLink}
\item $p_{i+1} = \mathtt{select}_{c}(\BWT_T,p_i-C[c])$ and $q_{i+1} = p_{i+1}+q_i-p_i$ for $i \in [1..k-1]$, where $c=W[i]$ is the character that satisfies $C[c] < p_i \leq C[c+1]$. This can be computed in $O(\log{\log{n}})$ time using a predecessor data structure that uses $O(\sigma)$ words of space \cite{Wi83}.\label{point:suffixLink}
\item Let $c \in [0..\sigma]$, and let $\INTERVAL{W[i..m]c}=[x_i..y_i]$ for $i \in [1..k]$. Then, $x_i = p_i+x_1-p_1$ and $y_i = p_i+y_1-p_1$.\label{point:child}
\end{enumerate}
\end{property}

The final property we will exploit relates the equivalence class of a maximal repeat to the equivalence classes of its in-neighbors in the CDAWG:

\begin{property}\label{obs:inNeighbors}
Let $w$ be a node in $\CDAWG_T$ with $\ell(w) = W \in [1..\sigma]^m$, and let $\mathcal{S}_w=\{W[1..m]$, $\dots$, $W[k..m]\}$ be the right-maximal strings that belong to the equivalence class of node $w$. Let $\{v^1,\dots,v^t\}$ be the in-neighbors of $w$ in $\CDAWG_T$, and let $\{V^1,\dots,V^t\}$ be their labels. Then, $\mathcal{S}_w$ is partitioned into $t$ disjoint sets $\mathcal{S}_w^1,\dots,\mathcal{S}_w^t$ such that $\mathcal{S}_w^i = \{W[x^i+1..m],W[x^i+2..m],\dots,W[x^i+|\mathcal{S}_{v^i}|..m]\}$, and the right-maximal string $V^{i}[p..|V^i|]$ labels the parent of the locus of the right-maximal string $W[x^i+p..m]$ in $\ST_T$.
\end{property}
\begin{proof}
It is clear that the parent in $\ST_T$ of every right-maximal string in the equivalence class of node $w$ belongs to the equivalence class of an in-neighbor of $w$: we focus here just on showing that the in-neighbors of $w$ induce a partition on the equivalence class of $w$. Assume that the character that labels arc $\gamma=(v^i,w)$ in the CDAWG is $c$. Since arc $\gamma$ exists, we can factorize $W$ as $X^i V^i Y^i$, where $Y^{i}[1]=c$, and we know that no prefix of $V^i Y^i$ longer than $V^i$ is right-maximal, and that no suffix of $W$ longer than $|V^i Y^i|$ is left-maximal. Consider any suffix $V^{i}[p..|V^i|]$ of $V^i$ that belongs to the equivalence class of $V^i$: if $p>1$, then $W[|X^i|+p..m]$ is not left-maximal, thus $W[|X^i|+p..m]$ belongs to the equivalence class of $W$. Its prefix $V^{i}[p..|V^i|]$ is right-maximal, and no longer prefix is right-maximal. Indeed, assume that string $V^{i}[p..|V^i|]Z^i$ is right-maximal for some prefix $Z^i$ of $Y^i$. Since $V^{i}[p..|V^i|]$ is not left-maximal, then string $V^{i}[p..|V^i|]Z^i$ is not left-maximal either, and this implies that $V^i Z^i$ is right-maximal, contradicting the hypothesis. Thus, string $V^{i}[p..|V^i|]$ labels the parent of the locus of string $W[|X^i|+p..m]$ in $\ST_T$. If $p=1$ and $V^i Y^i$ is not left-maximal, the same argument applies. If $V^i Y^i$ is left-maximal, then $W=V^i Y^i$, and since no right-maximal prefix of $W$ longer than $V^i$ exists, we have that $V^i$ labels the parent of the locus of $W$ in $\ST_T$.
\end{proof}

Combining Properties \ref{obs:interval}, \ref{obs:equivalenceClassInBWT} and \ref{obs:inNeighbors}, we obtain the following result:

\begin{table}[t!]
%\scriptsize
\begin{center}
\begin{tabular}{|l|c|c|c|c|c|c|c|c|c|}
\hline
 & $\mathtt{stringDepth}$ & $\mathtt{isAncestor}$ & $\mathtt{parent}$      & $\mathtt{child}$      & $\mathtt{suffixLink}$ & $\mathtt{weinerLink}$ & $\mathtt{edgeChar}$ & $\mathtt{nLeaves}$ \\
 & $\mathtt{locateLeaf}$  &                       & $\mathtt{nextSibling}$ & $\mathtt{firstChild}$ &                       &                       &                     &                        \\
\hline
\hline
1 & $O(1)$                 & $O(1)$                & $O(\log{\log{n}})$     & $O(1)$		   & $O(\log{\log{n}})$    & $O(\log{\log{n}})$    & $O(\log{\log{n}})$  & $O(1)$                 \\
\hline
2 & $O(1)$                 &                       & $O(\log{\log{n}})$     & $O(1)$                & $O(1)$                &                       &                     &                        \\
\hline
\end{tabular}
\caption{Time complexities of two representations of $\ST_T$: with intervals in $\BWT_T$ (row 1) and without intervals in $\BWT_T$ (row 2).}
\label{tab:suffixTree}
\end{center}
\vspace{-1cm}
\end{table}

\begin{theorem}\label{thm:suffixTreeRepresentation}
Let $T \in [1..\sigma]^{n-1}\#$ be a string. There are two implementations of $\ST_T$ that take $O(\newe_T + \newel_T)$ %$O(\newe + \newel)$ 
words of space each, and that support the operations in Table \ref{tab:suffixTree} with the specified time complexities.
\end{theorem}
\begin{proof}
We build $\RLCSA_T$ and $\CDAWG_T$, and we annotate the latter as described in Theorem \ref{thm:locate_dawg}, with the only difference that arcs that connect a maximal repeat to the sink are annotated with character and length like all other arcs. We store in every node $v$ of the CDAWG the number $v.\mathtt{size}$ of right-maximal strings that belong to its equivalence class, the interval $[v.\mathtt{first}..v.\mathtt{last}]$ of $\ell(v)$ in $\BWT_T$, a linear-space predecessor data structure \cite{Wi83} on the boundaries induced on the equivalence class of $v$ by its in-neighbors (see Observation \ref{obs:inNeighbors}), and pointers to the in-neighbor that corresponds to the interval associated with each boundary. Finally, we add to the CDAWG all suffix links $(v,w)$ from $\ST_T$ such that both $v$ and $w$ are maximal repeats, and the corresponding explicit Weiner links.

We represent a node $v$ of $\ST_T$ as a tuple $\mathtt{id}(v)=(v',|\ell(v)|,i,j)$, where $v'$ is the node in $\CDAWG_T$ that corresponds to the equivalence class of $v$, and $[i..j]$ is the interval of $\ell(v)$ in $\BWT_T$. Thus, operation $\mathtt{stringDepth}$ can be implemented in constant time, and if $v$ is a leaf, the second component of $\mathtt{id}(v)$ is its starting position in $T$. Operation $\mathtt{isAncestor}$ can be implemented by testing the containment of the corresponding intervals in $\BWT_T$. To implement operation $\mathtt{suffixLink}$, we first check whether $|\ell(v)|=v'.\mathtt{length}-v'.\mathtt{size}+1$: if so, we take the suffix link $(v',w')$ from $v'$ and we return $(w',w'.\mathtt{length},w'.\mathtt{first},w'.\mathtt{last})$. Otherwise, we return $(v',|\ell(v)|-1,i',j')$, where $[i'..j']$ is computed as described in point \ref{point:suffixLink} of Property \ref{obs:equivalenceClassInBWT}. To implement $\mathtt{weinerLink}$ for some character $c$, we first check whether $|\ell(v)|=v'.\mathtt{length}$: if so, we take the Weiner link $(v',w')$ from $v'$ labeled by character $c$ (if any), and we return $(w',w'.\mathtt{length}-w'.\mathtt{size}+1,i',j')$, where $[i'..j']$ is computed by taking a backward step with character $c$ from $[v'.\mathtt{first}..v'.\mathtt{last}]$. Otherwise, we check whether $\BWT_{T}[i]=c$: if so, we return $(v',|\ell(v)|+1,i',j')$, where $[i'..j']$ is computed as described in point \ref{point:weinerLink} of Property \ref{obs:equivalenceClassInBWT}.

To implement $\mathtt{child}$ for some character $c$, we follow the arc $\gamma=(v',w')$ in the CDAWG labeled by $c$ (see Observation \ref{obs:inNeighbors}), and we return tuple $(w',|\ell(v)|+\gamma.\mathtt{right},i',j')$, where $[i'..j']$ is computed as described in point \ref{point:child} of Property \ref{obs:equivalenceClassInBWT}. To implement $\mathtt{parent}$ we exploit Property \ref{obs:equivalenceClassInBWT}, i.e. we determine the partition of the equivalence class of $v'$ that contains $v$ by searching the predecessor of value $|\ell(v)|$ in the set of boundaries of $v'$: this can be done in $O(\log{\log{n}})$ time \cite{Wi83}. Let $\gamma=(u',v')$ be the arc that connects to $v'$ the in-neighbor $u'$ associated with the partition that contains $v$: we return tuple $(u',|\ell(v)|-\gamma.\mathtt{right},i',j')$, where $i'=i-v'.\mathtt{first}+u'.\mathtt{first}$ and $j'=j+u'.\mathtt{last}-v'.\mathtt{last}$ as described in point \ref{point:child} of Property \ref{obs:equivalenceClassInBWT}. Operation $\mathtt{nextSibling}$ can be implemented in the same way.

We read the label of an edge $\gamma$ of $\ST_T$ in $O(\log{\log{n}})$ time per character (operation $\mathtt{edgeChar}$), by storing $\RLCSA_{\REV{T}}$ and the interval in $\BWT_{\REV{T}}$ of the reverse of the maximal repeat that corresponds to every node of the CDAWG. By removing from $\mathtt{id}(v)$ the interval of $\ell(v)$ in $\BWT_T$, we can implement $\mathtt{stringDepth}$, $\mathtt{child}$, $\mathtt{firstChild}$ and $\mathtt{suffixLink}$ in constant time, and $\mathtt{parent}$ and $\mathtt{nextSibling}$ in $O(\log{\log{n}})$ time.
\end{proof}
%We keep it as an open problem to show how to efficiently support $\mathtt{LCA}$ 
%queries as well as other important operations like level-ancestor queries. 

\begin{corollary}
Let $T \in [1..\sigma]^{n-1}\#$ be a string. There is an implementation of $\ST_T$ that takes $O(\newe_T + \newel_T)$ %$O(\newe + \newel)$ 
words of space, that computes the matching statistics of a pattern $S \in [1..\sigma]^{m}$ with respect to $T$ in $O(m\log{\log{n}})$ time, and that can be traversed in $O(n\log{\log{n}})$ time and in a constant number of words of space.
\end{corollary}
\begin{proof}
We combine the implementation in the first row of Table \ref{tab:suffixTree} with the folklore algorithm for matching statistics, that issues $\mathtt{suffixLink}$ and $\mathtt{child}$ operations on $\ST_T$, and that reads the label of some edges of $\ST_T$. For traversal, we combine the implementation in the second row of Table \ref{tab:suffixTree} with the folklore algorithm that issues just $\mathtt{firstChild}$, $\mathtt{parent}$ and $\mathtt{nextSibling}$ operations.
\end{proof}

%The operations supported by the second representation of $\ST_T$ in Table \ref{tab:suffixTree} are enough to compute the matching statistics of a string $S \in [1..\sigma]^{m}$ with respect to $T$ in $O(m\log{\log{n}})$ time, 

By storing $\RLCSA_{\REV{T}}$ in addition to $\RLCSA_T$, and by adding to $\mathtt{id}(v)$ the interval of $\REV{\ell(v)}$ in $\BWT_{\REV{T}}$, we can also implement a bidirectional index on $T$ like those described in \cite{belazzougui2013versatile}, that supports the left and right extension of a string with any character in $O(\log{\log{n}})$ time and that takes $O(\newe + \newel)$ words of space.

\bibliographystyle{plain}
\bibliography{cpm2015}

\newpage

\section*{Appendix}

\subsection*{Lower bound on the number of arcs in the CDAWG}

\begin{lemma}\label{lemma:nArcsInCDAWG}
The number of arcs in $\CDAWG_T$ is $\Omega(\log{n})$ for any string $T \in [1..\sigma]^{n-1}\#$ and any $\sigma<n-1$.
\end{lemma}
\begin{proof}
$\CDAWG_T$ must contain a path from the source to the node that corresponds to every suffix of $T$, and since such paths are $n$, we need $\log{n}$ bits to discriminate at least one of these paths from the others. If $\sigma=2$, every node of the CDAWG has exactly two outgoing arcs, thus there must be a path from the source to the node associated with a suffix of $T$ that has length at least $\log{n}$. If $\sigma>2$, we can transform $\CDAWG_T$ into a DAG with degree at most two by multiplying the number of nodes and arcs by a factor of at most two. Indeed, if a node $v$ has outdegree $k$, we can replace the arcs that start from $v$ with a tree rooted at $v$ whose leaves are the original destinations of the arcs from $v$: this tree has $k-2$ additional nodes and $2k-2$ arcs. The DAG that results from this transformation must have at least $\log{n}$ arcs, thus the number of arcs in $\CDAWG_T$ is $\Omega(\log{n})$.
\end{proof}

The same proof clearly holds for \emph{left extensions of maximal repeats}, using $\CDAWG_{\REV{T}}$ rather than $\CDAWG_T$.

\subsection*{Proof of Lemma \ref{lemma:correspondence}}

\begin{proof}
Let $v$ be an internal node of $\ST_T$ such that $\ell(v)$ is a maximal repeat of $T$, and let $v'$ be the internal node of $\ST_{\REV{T}}$ such that $\ell(v')=\REV{\ell(v)}$. Then, for every edge $(v,w) \in \mathcal{F}^{r}$ in $\ST_T$ such that $v = \mathtt{parent}(w)$ there is an implicit Weiner link from $v'$ in $\ST_{\REV{T}}$ labeled by the first character of $\ell(v,w)$. Conversely, an implicit Weiner link labeled by character $b \in [0..\sigma]$ from any internal node $v'$ of $\ST_{\REV{T}}$ implies that $1 = |\Sigma^{r}_{\REV{T}}(b\ell(v'))| < |\Sigma^{r}_{\REV{T}}(\ell(v'))|$, therefore it must be that $|\Sigma^{\ell}_{\REV{T}}(\ell(v'))|>1$. It follows that $\REV{\ell(v')}$ is a maximal repeat of $T$, thus there is a node $v$ in $\ST_T$ with $\ell(v)=\REV{\ell(v')}$, and $b$ is the first character of the label of an edge $(v,w) \in \mathcal{F}^{r}$ such that $v = \mathtt{parent}(w)$.

Similarly, for every edge $(v,w) \in \mathcal{E}^{r}$ such that $v = \mathtt{parent}(w)$ there is an explicit Weiner link from $v'$ in $\ST_{\REV{T}}$ labeled by the first character of $\ell(v,w)$. Conversely, an explicit Weiner link labeled by character $b \in [0..\sigma]$ from any internal node $v'$ of $\ST_{\REV{T}}$ with at least two Weiner links implies that string $\REV{\ell(v')}$ is a maximal repeat, and that there is an edge $(v,w) \in \mathcal{E}^{r}$ such that $\ell(v)=\REV{\ell(v')}$, $v=\mathtt{parent}(w)$, and $\ell(v,w)=bV$ for some $V \in [0..\sigma]^*$.
\end{proof}

Lemma \ref{lemma:correspondence} immediately implies that the strings in $\mathcal{M}^{\ell}_T$ label internal nodes of $\ST_T$ that are not the destination of any suffix link. However, there can be internal nodes of $\ST_T$ that are not the destination of any suffix link but that are not maximal repeats.

\end{document}